\documentclass[reqno,12pt]{amsart}
\usepackage{amssymb,amsmath,amsthm,enumerate,amscd,graphicx,color}
\usepackage[usenames,dvipsnames,svgnames,table]{xcolor}
\usepackage{enumitem}

\usepackage{fullpage}
\usepackage{color}
 \usepackage[colorlinks=true, pdfstartview=FitV, linkcolor=blue, citecolor=blue, urlcolor=blue]{hyperref}
\theoremstyle{plain}

\newtheorem{corollary}{Corollary}[section]

\newtheorem{proposition}{Proposition}[section]

\theoremstyle{remark}

\newtheorem{remark}{Remark}[section]

\numberwithin{equation}{section}

\newcommand{\bC}{{\mathbb C}}

\newcommand{\bZ}{{\mathbb Z}}

\newcommand{\B}{\mathcal{B} }

\newcommand{\<}{\langle }

\renewcommand{\>}{\rangle}

\begin{document}

\title{Generalized  vertex operators of  Hall-Littlewood polynomials as twists of 
charged free fermions.}
\author{ G.\,Necoechea and N.Rozhkovskaya}

\begin{abstract}
Using  twists of fields of 
charged free fermions  we revise the generalized vertex operator presentation of Hall-Littlewood polynomials, propose  a new version of the  deformed boson-fermion correspondence, and construct 
 new examples of tau-functions  of the KP  hierarchy in  the  one-parameter deformation of the ring of symmetric functions $\Lambda[t]$.
\end{abstract}

\maketitle

\section{Introduction}

Vertex operator  realization of Hall-Littlewood polynomials through  generalized fermions  was first described in \cite{Jing1}, followed by the  construction of a deformed  version of boson-fermion correspondence in  \cite{Jing2}.  The ideas were  further developed and  applied in  works of many authors, such as \cite{Angul}, \cite{B-W}, \cite{D-E}, \cite{FW}, \cite{Sulk}. 

The first  goal of this note is to show that generalized fermions  in  \cite{Jing1} can be interpreted as a  simple twist of 
 charged free fermions that provide vertex operator realization of Schur symmetric functions  described in 
\cite{Jing3}, \cite {Md}, \cite{Zel}.

The deformed boson-fermion correspondence of \cite{Jing2}  establishes connection between the actions of generalized fermions and the twisted Heisenberg algebra.
This correspondence found its applications   in   \cite{FW}, \cite{Sulk}.  Our second goal is to propose a version  of boson-fermion correspondence that  relates  generalized fermions with the action of the classical Heisenberg algebra.  This construction is different from the deformed  boson-fermion correspondence  in \cite{Jing2}  and has its own advantages:  it does not require adjustments in the standard definition  of the normal ordered product of fields; the action of classical Heisenberg algebra and, as a consequence, of the Virasoro algebra are naturally present in the picture; the action of the  twisted Heisenberg algebra  appears as  a renormalization of the action of the classical Heisenberg algebra;   proofs of this version of the boson-fermion correspondence, the original  deformed boson-fermion correspondence of \cite{Jing2},  and its applications become simple implications of a twist of  the  classical boson-fermion correspondence \cite{DJKM83}.
   
The third result of this paper  is the construction of   tau-functions  of the KP hierarchy in the deformed ring of symmetric  functions 
$\Lambda[t]$. In \cite{Md}   symmetric functions  $S_\lambda$ are introduced as  a basis    dual to the basis of classical Schur functions  $s_\lambda$  with respect to the natural scalar form in the  deformed ring of symmetric functions  $\Lambda[t]$. Extending the idea of  twisting  the fields of  charged free fermions  we prove that $S_\lambda$'s are solutions of the bilinear KP identity and provide two different vertex operator realizations of this family of symmetric functions.  Note that this is stronger than  the result of  \cite{FW}  that  proves that  the  element $\Lambda[t]$ of the form  $\sum s_\lambda S_\lambda$  is a  tau-function  of the KP hierarchy.

 In Sections  \ref{s_prelim}, \ref{s_ferm} we review necessary facts about symmetric functions and the action of charged free fermions on the ring of symmetric functions. 
 In  Section \ref{s_twferm} we connect twisted fermions with classical charged free fermions. In Section \ref{s_bf} we propose a version of boson-fermion correspondence for twisted fermions and compare it with the construction
 of \cite{Jing2}. In Section \ref{s_tau} we prove that
   symmetric functions  $S_\lambda$ that  form a basis  of $\Lambda[t]$
  orthogonal to  the basis of Schur functions  $s_\lambda$  are  tau-functions of the KP hierarchy. We provide two different versions of their vertex operator realization.

\section{Preliminaries on symmetric functions}\label{s_prelim}
We  review necessary facts about symmetric functions following  \cite{Md}, \cite{Stan}.
Let $\Lambda= \Lambda [x_1,x_2,\dots] $ be the  ring of symmetric functions in  variables $ (x_1,x_2,\dots)$.
{\it  Schur symmetric   functions }$s_\lambda$  labeled by partitions $\lambda=(\lambda_1\ge \lambda_2\ge \dots\ge \lambda_n\ge 0)$ 
 defined by
 \[
s_\lambda(x_1,x_2,\dots )=\frac {\det[x_i^{\lambda_j+n-j}]}{\det[x_i^{n-j}]}
\]
 form a linear basis of $\Lambda$.
{\it Complete symmetric functions} $h_r= s_{(r)}$  are given by
\begin{align*}
h_r(x_1,x_2\dots)=\sum_{1\le i_1\le \dots \le i_r<\infty} x_{i_1}\dots x_{i_r},
\end{align*}
 {\it elementary symmetric functions} $e_r= s_{(1^r)}$  by
\begin{align*}e_r(x_1,x_2\dots)=\sum_{1< i_1< \dots < i_r<\infty} x_{i_1}\dots x_{i_r},
\end{align*}
 and {\it power sums} $p_k$ by
\[
p_k(x_1,x_2,\dots)=\sum_i x_i^k.
\]
It will be convenient  to set 
$h_{-k}(x_1,x_2\dots)=e_{-k}(x_1,x_2\dots)=p_{-k}(x_1,x_2\dots)=0
$
 for 
$k>0$ and $h_0=e_0=p_0=1$.
Each of these families generate the   ring of  symmetric functions $\Lambda$   is  a polynomial ring:
 \begin{align}
 \Lambda=\bC[h_1, h_2,\dots]=\bC[e_1, e_2,\dots]=\bC[p_1,p_2,\dots].
 \end{align}
 
Based on the interpretation of a polynomial ring as the ring of symmetric functions, one  defines {\it boson Fock space} $\B$. 
Let $\B=\bC[z, z^{-1}, p_1,p_2,\dots ]$ be the graded   space of polynomials 
\[
\B=\oplus_{m\in \bZ} \B^{(m)},\quad  \B^{(m)}= z^m\cdot \bC[ p_1,p_2,\dots ] = z^m \Lambda. 
\]
 
We  write generating  functions  for  complete, elementary symmetric functions and power sums:
 \begin{align}
&H(u)=\sum_{k\ge 0}  \frac{h_k}{u^k},\quad 
E(u)=\sum_{k\ge 0}\frac{ e_k}{ u^k},\quad P(u)=\sum_{k=1}^\infty p_k u^k. \label{e:22}
 \end{align}
 Then 
 \begin{align*}
H(u)= \prod_{i\ge 1} \frac{1}{1-x_i /u},\quad
E(u)=  \prod_{i\ge 1} (1+x_i /u),
 \end{align*}
\begin{align}\label{HE}
H(u) E(-u)=1,
\end{align}
and
\begin{align}\label{HEexp}
&H(u)= exp\left(\sum_{n\ge1} \frac{p_n}{n}\frac{1}{u^n}\right),\quad 
E(u)= exp\left(-\sum_{n\ge1} \frac{(-1)^{n} p_n}{n}\frac{1}{u^n}\right).
\end{align}
{\it Heisenberg algebra} is a complex  Lie algebra generated by  elements $\{\alpha_m\}_{m\in \bZ}$ and central element $1$  with commutation relations 
\begin{align}\label{def_alpha}
[\alpha_k,\alpha_n]=k\delta_{k,-n}\cdot 1.
\end{align}
Combining generators  into formal distribution $\alpha(u)=\sum_{k}\alpha_k{u^{-k-1}}$ we can rewrite this 
relation 
as
\begin{align}\label{heis}
[\alpha(u),\alpha(v)]=\partial_v\delta(u,v),
\end{align}
where $
\delta(u,v)= \sum_{k\in \bZ} {u^k}{v^{-k-1}}
$ is the formal delta distribution.
There is a  natural action of  Heisenberg algebra  on the  graded components $\alpha_n: \B^{(m)}\to  \B^{(m)}$ defined by  multiplication and 
differentiation operators:
\begin{align}\label{heis1}
\begin{cases}
\alpha_{-n}=p_n/n,  &\quad n> 0,\\
\alpha_{n}=\frac{\partial }{\partial p_n}, & \quad n>0,\\
\alpha_0=m. \quad 
\end{cases}
\end{align}

The
ring of symmetric functions $\Lambda$  possesses a natural  scalar product, where  classical Schur functions $s_\lambda$ constitute
  an orthonormal basis:
  \begin{align}\label{scalar}
  \<s_\lambda,s_\mu\>=\delta_{\lambda, \mu}.
  \end{align}
   Then for the  operator of  multiplication by a symmetric function $f$ one can
define the adjoint operator $f^\perp$ acting on the ring of symmetric functions by the standard rule:
$ \<f^\perp g, w\>=\<g,fw\>$, where $g,f,w\in \Lambda$.
We consider generating functions  of the adjoint operators
  \begin{align*}
  E^\perp(u)= \sum_{k\ge 0} {e^\perp_k} u^k,\quad H^\perp(u)= \sum_{k\ge 0} {h^\perp_k} u^k.
  \end{align*}
One can prove  
that 
  \begin{align}\label{DHEexp}
  E^\perp(u)= exp\left(-\sum_{k\ge 1} (-1)^k \frac{\partial}{\partial p_k} u^k\right),\quad\quad 
    H^\perp(u)= exp\left( \sum_{k\ge 1}\frac{\partial}{\partial p_k} u^k\right).
  \end{align}
The following commutation relations  serve as the foundation of most  of calculations in this note.
\begin{proposition}\label{comut1}\cite{Md}
\begin{align*}
\left(
1-\frac{u}{v}
\right)E^\perp(u)E(v)= E(v)E^\perp(u),\\
\left(
1-\frac{u}{v}
\right)H^\perp(u)H(v)= H(v)H^\perp(u),\\
H^\perp(u)E(v)= \left(
1+\frac{u}{v}
\right)E(v)H^\perp(u),\\
E^\perp(u)H(v)= \left(
1+\frac{u}{v}
\right)H(v)E^\perp(u).
\end{align*}
\end{proposition}
\begin{remark}\label{remark}
Statements of Proposition \ref{comut1}  should be understood as equalities of series expansions  in  powers of $u ^kv^{-m}$ for  $k, m\ge 0$. 
\end{remark}

\section{Fermions and Schur symmetric functions}\label{s_ferm}

Following \cite{DJKM83},  \cite{FLM}, \cite{K}, define  the action of algebra of charged free fermions on the boson Fock space.

Let  $R(u): \B\to \B$ act on  elements $z^m f$, $f \in \Lambda$ by the rule
\[
R(u) (z^mf)=\left(\frac{ z}{u}\right)^{m+1} f.
\]
(see e.g.\cite{K}, \cite{Bomb}).
Then  $R^{-1}(u)$ acts as
\[
R^{-1}(u) (z^mf)=\left({ z}\right)^{m-1}{{u ^m}} f.
\]
One should think of  $R^{\pm 1}(u)$ as operators that transport the action of other operators  along the  grading of  the boson Fock space:  $R^{\pm 1}(u): \B^{(m)}\to \B^{(m\pm 1)}$.
We set
\begin{align}
\Phi^+(u)&=uR(u)H(u) E^{\perp}(-u),\label{deff1}\\
\Phi^-(u) &=R^{-1}(u)E(-u) H^{\perp}(u).\label{deff2}
\end{align}
Observe that 
\begin{align*}
\Phi^+(u)|_{\B^{(m)}}&=zu^{-m}H(u) E^{\perp}(-u),\\
\Phi^-(u)|_{\B^{(m)}}&=z^{-1}{u^{m}}E(-u) H^{\perp}(u).
\end{align*}

\begin{proposition}\label{ferm}
Quantum fields $\Phi^\pm (u)$ satisfy the relations of the algebra of charged free fermions:
\begin{align}
\Phi^\pm(u)\Phi^\pm(v)+ \Phi^\pm(v)\Phi^\pm(u)&=0,\label{rel1}
\\
\Phi^+(u)\Phi^-(v)+ \Phi^-(v)\Phi^+(u)&=\delta(u,v).\label{rel2}
\end{align}
Here
$
\delta(u,v)= \sum_{k\in \bZ} \frac{u^k}{v^{k+1}}
$
is formal delta distribution.
\end{proposition}

\begin{proof}
We use Proposition \ref{comut1} to prove  this classical result, thus illustrating  the simplicity of this approach.   
Relations between other vertex operators further in this text are proved along the same lines. 

Using commutation relations of  Proposition \ref{comut1}, for any $f\in \Lambda$,
\begin{align*}
\Phi^+(u)\Phi^+(v) (z^m f)%z^{m+2} u^{-m-2}v^{-m-1} H(u) E^{\perp}(-u)H(v) E^{\perp}(-v)(f)\\
%&=z^{m+2} u^{-m-2}v^{-m-1} \left(1-\frac{u}{v}
%\right)H(u)H(v) E^{\perp}(-u) E^{\perp}(-v)(f)\\
&=z^{m+2} (uv)^{-m-2} \left(v-{u}
\right)H(u)H(v) E^{\perp}(-u) E^{\perp}(-v)(f).
\end{align*}
\begin{align*}
\Phi^-(u)\Phi^-(v)(z^m f)%&=z^{m-2}u^{m-2}v^{m-1}E(-u) H^{\perp}(u)E(-v) H^{\perp}(v)(f)\\
%&=z^{m-2}u^{m-2}v^{m-1}\left(1-\frac{u}{v}
%\right)E(-u)E(-v) H^{\perp}(u) H^{\perp}(v) (f)\\
&=z^{m-2}(uv)^{m-2}\left(v-{u}
\right)E(-u)E(-v) H^{\perp}(u) H^{\perp}(v) (f).
\end{align*}
Changing the roles of $u$  and $v$ in these calculations one gets (\ref{rel1}). 

For (\ref{rel2}) observe that 
\begin{align*}
 \left(1-\frac{u}{v}
\right)\Phi^+(u)\Phi^-(v)(z^m f)
=z^mu^{-m}v^{m-1}H(u)E(-v) E^{\perp}(-u) H^{\perp}(v)(f),
\end{align*}

\begin{align*}
 \left(1-\frac{v}{u}
\right)\Phi^-(v)\Phi^+(u)(z^m f) 
=z^mu^{-m-1}v^{m}H(u)E(-v) E^{\perp}(-u) H^{\perp}(v)(f).
\end{align*}
Note that  by Remark \ref{remark},
\[
 \left(1-\frac{u}{v}
\right)^{-1}\sum_{k\ge 0} \frac{u^{k}}{v^{k+1}},\quad  \left(1-\frac{v}{u}
\right)^{-1}=\sum_{k\ge 0} \frac{v^{k}}{u^{k+1}}.
\]
Then
%Denote as $i_{u/v}  ( F)$  the expansion of rational function $F$ as a series in powers  of $u/v$. Then  one can write
\begin{align*}
&(\Phi^+(u)\Phi^-(v)+\Phi^-(v)\Phi^+(u)) (z^mf)\\
%&= z^m\frac{v^m}{u^{m}}\left(\frac{1}v\, i_{u/v}  \left(1-\frac{u}{v}
%\right)^{-1} + \frac{1}{u}\,  i_{v/u} \left(1-\frac{v}{u}
%\right)^{-1} \right )H(u)E(-v) E^{\perp}(-u) H^{\perp}(v) (f)
%\\
&=z^m\frac{v^m}{u^{m}}\left(
\sum_{k\ge 0} \frac{u^{k}}{v^{k+1}} +\sum_{k\ge 0} \frac{v^{k}}{u^{k+1}}
\right)H(u)E(-v) E^{\perp}(-u) H^{\perp}(v) (f)
\\
&=z^m\frac{v^m}{u^{m}} \,\delta(u, v)H(u)E(-v) E^{\perp}(-u) H^{\perp}(v) (f)
=\delta(u, v)\cdot z^m f.
\end{align*}
We used    (\ref{HE}) along with  the property  of formal delta distribution $
\delta(u, v) A(v)=\delta(u, v) A(u) 
$
for any formal distribution $A(u)$
 (see  e.g. \cite{K}, \cite{Bomb}).
 \end{proof} 
 
Let $1\in \B^{(0)}$ be the constant function.
%%%%%%%%%%%%%%%%
\begin{proposition}\label{v_pres_s}
\[
\Phi^+(u_1)\dots \Phi^+(u_l) \,  (1)= z^l u_1^{-l} \dots u_l^{-1}Q(u_1,\dots, u_l),
\]
where 
\[
Q(u_1,\dots, u_l)=\prod_{1\le i<j\le l}\left(1-\frac{u_i}{u_j}\right)\prod_{i=1}^{l} H(u_i).
\]
\end{proposition}
\begin{proof}
Using  Proposition \ref{comut1} and that $E^\perp(-u) (1)=1$,
 we write: 
\begin{align*}
\Phi^+(u_1,)\dots \Phi^+(u_l) \,  (1)&= z^l u_1^{-l} \dots u_l^{-1}\prod_{1\le i<j\le l}\left(1-\frac{u_i}{u_j}\right)\prod_{i=1}^{l} H(u_i) E^\perp(-u_i) (1)
\\
&=z^l u_1^{-l} \dots u_l^{-1}Q(u_1,\dots, u_l).
\end{align*}

\end{proof}
 It is known \cite{Md},  \cite{Jing3}  that $Q(u_1,\dots, u_l)$  is the generating function for Schur symmetric functions in the following sense.
Consider  the  series expansion  of the rational function 
 \[
Q(u_1,\dots, u_l)=\sum_{(\alpha_1,\dots, \alpha_l)\in \bZ^l} Q_\lambda \,u_1^{-\lambda_{1}} \dots u_{l}^{-\lambda_{l}}
 \]
 in the region  $|u_1|<\dots <|u_l|$. Then for any partition $\lambda=(\lambda_1\ge \dots \ge \lambda_l\ge 0)$ the coefficient 
of  $u_1^{-\lambda_{1}}\dots  u_{l}^{-\lambda_{l}} $ 
is exactly Schur symmetric  function: $Q_{(\lambda_1,\dots, \lambda_l)}= s_\lambda$.
Thus, Proposition \ref{v_pres_s} describes vertex operator presentation of Schur symmetric functions  (cf. \cite{Jing3}, \cite {Md}, \cite{Zel}).

\section{Generalized fermions and Hall-Littlewood  symmetric functions. }\label{s_twferm}

Let $\lambda$  be a  partition  of length  at most $n$, let $t$ be a parameter.  Hall-Littlewood  polynomials  are defined  by
\begin{align*}
P_\lambda (x_1,\dots, x_n;t)=\left(\prod_{i\ge 0}\prod_{j=1}^{m(i)}\frac{1-t}{1-t^j}\right)\sum_{\sigma \in S_n}
\sigma\left(x_1^{\lambda_1}\dots x_n^{\lambda_n}
\prod_{i < j}\frac{x_i-tx_j}{x_i-x_j}\right),
\end{align*}
where $m(i)$  is  the number  of  parts of the partition $\lambda$  that are  equal to $i$, and
 $S_n$  is the symmetric group of $n$ letters   \cite{Md}. Labeled by partitions Hall-Littlewood polynomials  form a linear basis of  the deformed ring $\Lambda[t]$ of symmetric polynomials 
 with coefficients in $\bC[t]$.

In this section we  
show that  vertex operators presentation of 
Hall-Littlewood polynomials   $P_\lambda$ is obtained  from vertex operators  presentation of Schur symmetric functions $s_\lambda$ by a simple twist of the fields of charged free fermions by multiplication by $E(-u/t)$ or $H(u/t)$.
 This approach significantly simplifies the technical  proofs of  \cite{Jing1}, \cite{Jing2} and  provides new insight into  the original  results of these  papers. 
 
Consider the  deformed  boson Fock space $\B(t)= \oplus \B^{(m)}[t]$, where $ \B^{(m)}[t]= z^m\Lambda[t]$.
We  extend the  action of the operators in Section 2 to   $\B(t)$ by $t$-linearity.
 Define quantum fields of operators  acting on  on $\B(t)$
\begin{align}
\Psi^+(u)&= E(-u/t) \Phi^+(u)\,= uR(u) E(-u/t)H(u) E^{\perp}(-u),\label{phps1}
\\
\Psi^-(u)&= H(u/t) \Phi^-(u)\,= R^{-1}(u)H(u/t)E(-u) H^{\perp}(u).\label{phps2}
\end{align}

\begin{proposition}\label{ferm_twisted}
 Quantum fields $\Psi^\pm(u) $ satisfy  relations of generalized fermions
\begin{align}
\left(1-\frac{ut}{v}\right)\Psi^\pm(u)\Psi^\pm(v)+ \left(1-\frac{vt}{u}\right)\Psi^\pm(v)\Psi^\pm(u)&=0,\label{twist1}
\\
\left(1-\frac{vt}{u}\right)\Psi^+(u)\Psi^-(v)+\left(1-\frac{ut}{v}\right)\ \Psi^-(v)\Psi^+(u)&=\delta(u,v)(1-t)^2.\label{twist2}
\end{align}
\end{proposition}

\begin{proof}
The proof  is based on commutation relations of  Proposition \ref{comut1} and  follows the same  lines as the proof of Proposition \ref{ferm}.
\end{proof}
Proposition \ref{ferm_twisted}  immediately implies  that  operators $\Psi^\pm(u)$  provide vertex operators realization \cite{Jing1} of 
Hall-Littlewood polynomials. Let  
\begin{align*}
F(u_1,\dots, u_l; t)&=\prod_{1\le i<j\le l}\frac{u_j-u_i}{u_j-u_it}\prod_{i=1}^{l} H(u_i)E(-u_i/t),
\end{align*}
where    the  expression  $\prod_{1\le i<j\le l}\frac{u_j-u_i}{u_j-u_it}$  is understood as the series expansion of this rational function 
in the region $|u_1|<\dots <|u_l|$.
Consider the 
 expansion 
 \[
F(u_1,\dots, u_l; t)=\sum_{\lambda_1,\dots \lambda_l\in \bZ} F_\lambda\, u_1^{-\lambda_{1}} \dots u_{l}^{-\lambda_{l}}. 
 \]
 in the region  $|u_1|<\dots < |u_l|$. 
It is proved in    \cite{Jing1}  (see also \cite{Md})  that  for any partition $\lambda=(\lambda_1\ge \dots \ge \lambda_l\ge 0)$ the 
coefficient of  $u_1^{-\lambda_{1}}\dots  u_{l}^{-\lambda_{l}} $ 
is exactly Hall-Littlewood symmetric function: $F_\lambda= P_\lambda (x_1,x_2,\dots ; t)$.
\begin{proposition}\label{HL-vert} (\it{Vertex operator presentation of Hall-Littlewood symmetric functions})
One has
\begin{align}\label{VpHL}
\Psi^+(u_1)\dots \Psi^+(u_l) \,  (1)=  z^lu_1^{-l} \dots u_l^{-1}F(u_1,\dots, u_l; t).
\end{align}
\end{proposition}
\begin{proof} From definition  (\ref{phps1}), (\ref{phps2}) and  commutation relations of  Proposition \ref{comut1}
one immediately finds that 
\begin{align*}
&\Psi^+(u_1)\dots \Psi^+(u_l)(1) \, \\
&= z^l \prod_{1\le i<j\le l}
\left(1-\frac{u_i}{u_j}\right)\left(1-\frac{u_it}{u_j}\right)^{-1} 
 \prod_{i=1}^{l}  u_i^{i-l-1}H(u_i)E(-u_i/t)\prod_{i=1}^{l} E^\perp(-u_i) (1),
 \end{align*}
 which simplifies to  $ z^lu_1^{-l} \dots u_l^{-1}F(u_1,\dots, u_l; t)$ since $E^\perp(-u) (1)=1$.
\end{proof}

\begin{corollary}
Generating functions  $F(u_1,\dots, u_l; t)$ for Hall-Littlewood polynomials  and 
$Q(u_1,\dots, u_l)$ for Schur symmetric functions are related by the formula
\[
F(u_1,\dots, u_l; t)=\prod_{1\le i<j\le l}\left( 1-\frac{tu_i}{u_j}\right)^{-1}\prod_{i=1}^{l} E(-u_i/t)\, Q(u_1,\dots, u_l).
\]
\end{corollary}
%%%%%%%%%%%%%%%
\section{Boson-fermion correspondence for Hall-Litllewood polynomials revisited}\label{s_bf}

The classical  boson-fermion correspondence connects   the action of  charged free fermions and the action of  (classical) Heisenberg algebra   with generators
$\{\alpha_m\}_{m\in \bZ}$,  the central element $1$, and  relations (\ref{def_alpha}):
\begin{center}
(I) \quad  Heisenberg algebra $\to$  charged free fermions $\to$ Heisenberg algebra.
\end{center}
In  the same spirit  a deformed boson-fermion correspondence between the actions of generalized fermions and twisted Heisenberg algebra  was established in  \cite{Jing2}:
\begin{center}
 (II) \quad twisted Heisenberg algebra $\to$  generalized  fermions $\to$  twisted Heisenberg algebra.
\end{center}
The twisted Heisenberg algebra is defined as  an algebra with generators  $\{\bold{h}_k\}$, central element $c$, and relations
\begin{align}\label{twist_h}
[\bold{h}_n, \bold h_m]=\frac{m\delta_{m,-n} }{1-t^{|m|}}\cdot c.
\end{align}
This construction found its applications  in \cite{FW},  \cite{Sulk}.  At the same time, it has certain disadvantageous  deviations from the format of the classical  boson-fermion correspondence. In particular,  to obtain the  bosonisation   \cite{Jing2}  of generalized fermions   one has to change the standard definition of normal ordered product of fields. 
 Moreover, the  natural presence of the action of the  classical Heisenberg algebra and the Virasoro  algebra is not reflected by this deformed version. 

In this section we  propose  another deformed construction of boson-fermion correspondence,   different from \cite{Jing2}.   It establishes  connection between the actions of generalized fermions 
and  classical Heisenberg algebra: 
\begin{center}
(III) \quad  Heisenberg algebra $\to$   generalized fermions $\to$ Heisenberg algebra.
\end{center}
Among the advantages  of  the boson-fermion correspondence   (III)  over  (II)  is that 
(a) the standard definition   of the normal ordered product of fields  is used   in all definitions;  (b) the action of the twisted Heisenberg algebra (\ref{twist_h})
 is present  as certain renormalization of the  action of the classical Heisenberg algebra;  (c)  the action of  Virasoro  algebra remains  naturally  in the construction;
 (d)  proofs  of statements of  correspondences (II)  and (III) become  simple implications of the results of classical boson-fermion correspondence (I). 

We recall  the statement of    the classical boson-fermion correspondence in the form convenient  for our exposition. 
For  a formal distribution $a(z)=\sum_{n\in \bZ} a_nz^{-n-1}$ 
define
$
a(z)_+=\sum_{n\le -1} a_nz^{-n-1}, \quad a(z)_-=\sum_{n\ge 0} a_nz^{-n-1}.
$
The normal ordered  product of two formal distributions  is the formal distribution defined  by the formula
$
:a(z) b(z):= a(z)_+ b(z)+ b(z) a_(z)_-
$
(see e.g.  \cite{Bomb} ).

\begin{proposition} \label{BI}(\it{Classical Boson-fermion correspondence (I)}, \cite{DJKM83},\cite{FLM}, \cite{K}, )
\begin{enumerate}[label=\alph*)]
\item 
 Consider the action (\ref{heis1}) of the Heisenberg algebra  on the boson space $\B$. Then
the fields 
\begin{align}
   %  \Phi^+(u)= uR(u)\exp \left(\sum_{n\ge 1}\frac{p_n}{n} \frac{1}{u^n}  \right) \exp \left(-\sum_{n\ge 1}\frac{\partial} {\partial p_n} {u^{n}}\right),
          \Phi^+(u)= uR(u)\exp \left(\sum_{n\ge 1}\alpha_{-n} \frac{1}{u^n}  \right) \exp \left(-\sum_{n\ge 1}\alpha_n {u^{n}}\right),
\label{phi2}\\
   %  \Phi^-(u)=R^{-1}(u)\exp \left(-\sum_{n\ge 1}\frac{p_n}{n} \frac{1}{u^n}  \right) \exp \left(\sum_{n\ge 1}\frac{\partial} {\partial p_n} {u^{n}}\right).
      \Phi^-(u)=R^{-1}(u)\exp \left(-\sum_{n\ge 1}\alpha_{-n} \frac{1}{u^n}  \right) \exp \left(\sum_{n\ge 1}\alpha_n {u^{n}}\right)
\label{phi21}
\end{align}
satisfy  relations  (\ref{rel1}), (\ref{rel2}) and define the action of   charged free fermions on  $\B$.

\item 
  Let $ \Phi^\pm(u)=\sum_{k\in \bZ} \Phi^\pm_{k+{1}/2} u^{\pm k}$ satisfy  relations  (\ref{rel1}), (\ref{rel2}).
Introduce 
\[
\Phi^+(u)\,_+=\sum_{k\ge 1} \Phi^+_{k+1/2}u^{k},\quad \Phi^+(u)\,_-=\sum_{k\le 0} \Phi^+_{k+1/2}u^{k}.
\]
Then coefficients of  the formal distribution 
\[
\alpha(u)=:\Phi^+(u)\cdot\Phi^-(u):\, =\Phi^+(u)\,_+\,\Phi^-(u)-\Phi^-(u)\,\Phi^+(u)\,_-
\]
satisfy the relations (\ref{def_alpha}) of  Heisenberg algebra.
\item  Let $ \Phi^\pm(u)=\sum_{k\in \bZ} \Phi^\pm_{k+{1}/2} u^{\pm k}$ satisfy  relations  (\ref{rel1}), (\ref{rel2}).
 For any $\beta\in \bC$,  the coefficients of the formal distribution 
$
 L^{(\beta)} (u)=\sum_{k\in \bZ} L_k u^{-k-2}
 $
 defined by the formula
 \[
  L^{(\beta)} (u)= \beta :\partial\Phi^+(u)\,\Phi^-(u):+ (1- \beta) :\Phi^+(u)\partial\Phi^-(u):
 \]
satisfy the  relations of  Virasoro  algebra  with central charge $c_\beta =-12\beta^2 +12\beta -2$:
 \[
  [L^{(\beta)} (u), L^{(\beta)} (v)] =\partial_v L(v)\,\delta(u,v) + 2L(v)\, \partial_v\delta(u,v) + \frac{c_\beta}{12}\,\partial^3_v \delta(u,v). 
 \] 
\end{enumerate}
\end{proposition}
\begin{proof}
Given the action (\ref{heis1}) of the Heisenberg algebra, we use presentations (\ref{HEexp}), (\ref {DHEexp})   to compare (\ref{phi2}) and (\ref{phi21}) 
with the definition (\ref{deff1}) and (\ref{deff2}) of $\Phi^{\pm}(u)$ to get  part  a), while for the proofs of b) and c) we refer to \cite{Bomb}.
\end{proof}

%%%%%%%%%%%%%%%
\begin{proposition} (\it{Revisited  version Boson-fermion correspondence (III)})
\begin{enumerate}[label=\alph*)]

\item 
The action of  the twisted Heisenberg algebra  can be defined  by renormalization of the action  of  the classical  Heisenberg algebra:
\begin{align}\label{heis3}
\begin{cases}
\bold{h}_{-n}=\alpha_{-n},  &\quad n\ge 0,\\
\bold{h}_n=\frac{1}{1-t^n}\alpha_{n}& \quad n>0,
\quad c=1.
\end{cases}
\end{align}
\item 
Consider the  action (\ref{heis1})  of  Heisenberg algebra on the boson space $\B  $ expanded to the  deformed space $\B(t)$
by the rule $\alpha_n(t^sf)= t^s\alpha_n(f)$. Then the fields
\begin{align}
    % \Psi^+(u)= R(u)\exp \left(\sum_{n\ge 1}\frac{1-t^n}{n}p_n \frac{1}{u^n}  \right) \exp \left(-\sum_{n\ge 1}\frac{\partial} {\partial p_n} {u^{n}}\right),
      \Psi^+(u)= uR(u)\exp \left(\sum_{n\ge 1}{(1-t^n)}\alpha_{-n} \frac{1}{u^n}  \right) \exp \left(-\sum_{n\ge 1}\alpha_n {u^{n}}\right),
\label{psi2}\\
    % \Psi^-(u)=R^{-1}(u)\exp \left(-\sum_{n\ge 1}\frac{1-t^n}{n} p_n\frac{1}{u^n}  \right) \exp \left(\sum_{n\ge 1}\frac{\partial} {\partial p_n} {u^{n}}\right).
     \Psi^-(u)=R^{-1}(u)\exp \left(-\sum_{n\ge 1}(1-t^n) \alpha_{-n}\frac{1}{u^n}  \right) \exp \left(\sum_{n\ge 1}\alpha_n {u^{n}}\right)
\label{psi21}
\end{align}
satisfy relations   (\ref{twist1}), (\ref{twist2}) and define the action of  generalized  fermions  on $B(t)$.

\item  The other way, consider  the action of  generalized  fermions $\Psi^\pm(u)$ on $B(t)$ defined by 
(\ref{phps1}), (\ref{phps2}).
Then the coefficients  of the normal ordered product 
\[
\alpha (u) =:  H(u/t) \Psi^+(u)\,\cdot\, E(-u/t) \Psi^-(u):
\]
satisfy  relations (\ref{def_alpha}) and define the action of the classical  Heisenberg algebra on $\B(t)$.

\item

Given  the action  (\ref{phps1}), (\ref{phps2}) of  generalized  fermions $\Psi^\pm(u)$ on $\B(t)$,  the coefficients of the formal  distribution 
$
 L^{(\beta)} (u)=\sum_{k\in \bZ} L_k u^{-k-2}
 $
 defined by the formula
  \begin{align*}
  L^{(\beta)} (u)&= \beta : H(u/t)\left( t^{-1}P(-u/t) \Psi^+(u) + \partial \Psi^+(u)\right) \cdot \, E(-u/t)\Psi^-(u):\quad\\
   &+ (1- \beta)  : H(u/t)\Psi^+(u) \cdot  E(-u/t)\left( -t^{-1}P(-u/t) \Psi^-(u) + \partial \Psi^-(u)\right):.
 \end{align*}
 satisfy the  relations of  the  Virasoro  algebra  with central charge $c_\beta =-12\beta^2 +12\beta -2$.
 \end{enumerate}
\end{proposition}

\begin{proof}
 Direct check of  commutation relations proves part  a). 
Given the action (\ref{heis1}) of the Heisenberg algebra we use presentations (\ref{HEexp}), (\ref {DHEexp})   to compare (\ref{psi2}), (\ref{psi21}) 
with  definition (\ref{phps1}), (\ref{phps2}) of $\Psi^{\pm}(u)$ to obtain part b).
 Note  that from part a) and  (\ref{phps1}), (\ref{phps2}) quantum 
fields 
\begin{align}\label{subst}
\Phi^+(u)= H(u/t) \Psi^+(u),\quad 
\Phi^-(u)= E(-u/t) \Phi^-(u)
\end{align}
satisfy  relations of charged free  fermions. Then   part c)  follows from  Proposition \ref{BI} part b). 

Part d) 
follows from the substitution of   (\ref{subst})  into  the formula of $L^{(\beta)}(u)$ in Proposition \ref{BI} part c)
and  the property \cite{Md}
$$
P(u)= \partial H(u)/H(u), \quad P(-u)= \partial E(u)/E(u).$$

\end{proof}

\section{ Tau-functions of the KP hierarchy in $\Lambda[t]$ }\label{s_tau}
%%%%%%%%%%%%%%%%%%%%%%%%%%%%

The ring $\Lambda[t]$ possesses a scalar product
$\<\cdot ,\cdot \>_t$, which is a deformation of the scalar product (\ref{scalar})  on $\Lambda$.  
Following  III.4 in \cite{Md}, define  symmetric functions  $S_{\lambda}= 
S_{\lambda}(x_1,x_2,\dots ; t) $  as a
basis dual  to classical Schur functions $s_\lambda =s_\lambda(x_1,x_2,\dots)$ with respect to  the deformed scalar product: 
 \[
 \<S_\lambda,s_\mu\>_t=\delta_{\lambda, \mu}.
 \]

\begin{proposition}
Let
$$
S(u_1,\dots, u_l)= \prod_{i<j}{\left(1-\frac{u_i}{u_j}\right)}\prod_{i=1}^{l}H(u_i) E(-u_i/t).
$$
Then  for any partition $\lambda$, the coefficient of $u_1^{-\lambda_1}\dots u_l^{-\lambda_l}$
is  $S_\lambda$.

\end{proposition}
\begin{proof}
By  III.4  formula (4.3) and III.2  formula (2.10) in  \cite{Md}, 
 $S_\lambda$  can be expressed through Hall-Littlewood polynomials  $P_{(k)}$ by an analogue of Jacobi\,-\,Trudi formula: 
\begin{align}\label{Sq}
S_\lambda  =\det [(1-t)P_{(\lambda_i-i+j)}].
\end{align}
Here  $P_{(k)}= P_{(k)}(x_1,x_2,\dots ;t)$  is  the coefficient of the expansion
 \[
 S(u;t)= H(u) E(-u/t)=(1-t) \sum_{k=0}^{\infty} P_{(k)}\frac{1}{u^k}.
 \]
Then
\begin{align*}
S(u_1,\dots, u_l)&= \prod_{i<j}{\left(1-\frac{u_i}{u_j}\right)}\prod_{i=1}^{l}S(u_i; t) \\
&
= \det[u_i^{-i+j}] \prod_i S(u_i; t)
= \det [u_i^{-i+j}S(u_i;t) ]\\ 
&=(1-t)^l\sum_{\sigma\in S_l} sgn(\sigma) \sum_{a_1\dots a_l} P_{(a_1)}u_1^{-a_1-1+\sigma(1)}\dots P_{(a_l)}u_l^{-a_l-l+\sigma(l)}\notag\\
&=(1-t)^l
\sum_{\lambda_1\dots \lambda_l} \sum_{\sigma\in S_l}sgn(\sigma)P_{(\lambda_1-1+\sigma(1))}\dots P_{(\lambda_l-l+\sigma(l))}u_1^{-\lambda_1}\dots u_l^{-\lambda_l}\notag
\\
&=\sum_{\lambda_1\dots \lambda_l}\det [(1-t)P_{(\lambda_i-i+j)}]u_1^{-\lambda_1}\dots u_l^{-\lambda_l}.
\end{align*}
\end{proof}
\begin{corollary} Generating functions  $S(u_1,\dots, u_l)$ for symmetric functions $S_\lambda$ and
$Q(u_1,\dots, u_l)$ for   Schur symmetric functions $s_\lambda$   are related by the formula
\begin{align*}
S(u_1,\dots, u_l)&=E(-u_l/t)\dots E(-u_1/t) Q(u_1,\dots u_l).
\end{align*}
The vertex operator presentation of  the generating function $S(u_1,\dots, u_l)$  can be written as
\begin{align*}
z^l u_1^{-l}\dots u_l^{-1} S(u_1,\dots, u_l)&=
E(-u_l/t)\dots E(-u_1/t)
\Phi^+(u_1)\dots \Phi^+(u_l) \,  (1).
\end{align*}
\end{corollary}

In \cite{FW}  the relation of symmetric functions  $S_\lambda$ to tau-functions of the KP hierarchy  is discussed. The  determinant-type property (\ref{Sq})
of symmetric functions   $S_\lambda$  is interpreted as 
 pl\" ucker coordinates-type property. This observation allows the authors  of   \cite{FW} to  conclude  that the expression   $\sum_{\lambda}s_\lambda S_\lambda$    is an example of   a  tau-function  of the KP hierarchy in $\Lambda[t]$.

Here we show that   symmetric functions  $S_\lambda$'s themselves are  tau-functions  of the KP hierarchy;  this result is not present in   \cite{FW} and does not follow from the mentioned above statement. 
Moreover, we provide  the explicit  formula for the  charged free fermions  action that realize   the KP hierarchy  for these tau-functions and conclude with one more vertex operator presentation of  the generating function  of $S_\lambda$'s.

As it is well-known, the  bilinear form of the KP hierarchy is the  equation 
\[
\Omega (\tau \otimes \tau)=0,
\]
where $\tau\in \B^{(0)}=\Lambda= \bC[p_1,p_2,\dots]$ and
\begin{align}\label{KP}
\Omega=\operatorname{Res}_{u=0}\left( \frac{1}{u}\Phi^+(u)\otimes \Phi^-(u)\right).
\end{align}
It is known  that Schur symmetric functions  $s_\lambda\in \B^{(0)}$ are solutions of the KP hierarchy  \cite{DJKM83}, \cite{JimMiw1}, \cite{JimMiw2}. 
Define  formally
\begin{align}
     \Phi_t^+(u)&= uR(u)H(u) E(-u/t)\prod_{i=0}^{\infty}E^\perp(-u/t^i),\label{phit1}\\
      \Phi_t^-(u)&= R(u)^{-1}E(-u) H(u/t)\prod_{i=0}^{\infty}H^\perp(u/t^i).\label{phit2}
\end{align}
Consder 
\begin{align}\label{omt2}
\Omega_t=\operatorname{Res}_{u=0}\left( \frac{1}{u}\Phi_t^+(u)\otimes \Phi_t^-(u)\right),
\end{align}
 and  the equation 
 \begin{align}\label{omt1}
\Omega_t (\tau \otimes \tau )=0, 
\end{align}
where $\tau\in \B^{(0)}(t)$.
We summarize the statements in the following proposition. 
\begin{proposition}
\begin{enumerate}[label=\alph*)]
\item  
Let $\Phi_t^\pm(u)$ be defined by  (\ref{phit1}), (\ref{phit2}) with the expansion in the region $|t|<1$.
  Operators $\Phi_t^\pm(u)$
satisfy exactly the same relations as classical charged free  fermions  $\Phi^\pm(u)$ in Proposition  \ref{ferm}.
 Thus, operators $\Phi_t^\pm(u)$ provide the action of  charged free fermions  on  the deformed space $\B(t)$,
Consequently, the equation  (\ref{omt1}) is  the bilinear identity of the KP hierarchy on functions $\tau \in \B(t)$.

\item   Symmetric functions $S_\lambda$  are solutions of the bilinear identity (\ref{omt1}).
 Consequently,  symmetric functions $S_\lambda$  are tau-functions of the KP hierarchy.
\item  Generating function $S(u_1,\dots, u_l)$ of symmetric functions  $S_\lambda$'s
has  the second   vertex operator presentation:
 \[
\Phi_t^+(u_1)\dots \Phi_t^+(u_l) \,  (1)= z^l u_1^{-l} \dots u_l^{-1}S(u_1,\dots, u_l).
\]
\end{enumerate}
\end{proposition}
\begin{proof}
Schur symmetric functions  $s_\lambda$ are expressed through complete symmetric functions by Jacobi\,-\,Trudi formula
$
s_\lambda= \det[h_{\lambda_i -i+j}] $, while 
 complete symmetric functions   $h_k$ in the determinant can be expressed as   polynomial functions of power sums $h_k=h_k(p_1, p_2,\dots )$
through the relation 
\[
\sum_{k=0}^{\infty} \frac{h_k }{u^k}= H(u)= exp\left(\sum_{n\ge1} \frac{p_n}{n}\frac{1}{u^n}\right).
\]

At the same time, from   (\ref{Sq}) the  dual symmetric functions $S_\lambda$ are given by 
$S_\lambda= \det[q_{\lambda_i -i+j}]$,
where symmetric functions 
$q_k= (1-t) P_{(k)}$ are expressed  as functions  of power sums  $q_k=q_k(p_1, p_2,\dots )$  through the relation
\[
\sum_{k=0}^{\infty} \frac{q_k }{u_k}=H(u) E(-u/t)= exp\left(\sum_{n\ge1} \left(1-t^{n}\right) \frac{p_n}{n}\frac{1}{u^n}\right).
\]
Thus, $S_\lambda$ as  a function of  power sums $(p_1, p_2,\dots )$   can be obtained from  $s_\lambda$ by the substitution   of variables $p_n\to (1-t^n)p_n$.

Recall  that Schur symmetric functions  $s_\lambda$ are solutions of the bilinear  KP identity (\ref{KP}). 
  $\Phi^\pm(u)$ in (\ref{KP})  in terms of  operators $p_i$'s  and $\partial/\partial  p_i$'s have  the form (\ref{phi2}), (\ref{phi21}).
Hence
  $S_\lambda$ satisfies  the bilinear identity
\begin{align*}
\Omega_t (S_\lambda \otimes S_\lambda)=0, 
\end{align*}
where
\begin{align}\label{omt2}
\Omega_t=\operatorname{Res}_{u=0}\left( \frac{1}{u}\Phi_t^+(u)\otimes \Phi_t^-(u)\right),
\end{align}
and  $\Phi^\pm_t(u)$  is obtained from (\ref{phi2}), (\ref{phi21}) by the same substitution $p_n\to (1-t^n)p_n$:
\begin{align*}
     \Phi_t^+(u)= uR(u) \exp \left(\sum_{n\ge 1}\frac{(1-t^n)p_n}{n} \frac{1}{u^n}  \right) \exp \left(\sum_{n\ge 1}\frac{1}{(1-t^n)}\frac{\partial} {\partial p_n} {u^{n}}\right).\\
         \Phi_t^-(u)= R(u)^{-1} \exp \left(-\sum_{n\ge 1}\frac{(1-t^n)p_n}{n} \frac{1}{u^n}  \right) \exp \left(\sum_{n\ge 1}\frac{1}{(1-t^n)}\frac{\partial} {\partial p_n} {u^{n}}\right).
\end{align*}
Using the geometric series expansion $1/(1-t^n)=\sum_{i} (t^i)^n$  in the region $|t|<1$ we can write the second exponential factor as 
\begin{align*}
\exp \left(\mp\sum_{n\ge 1}\frac{1}{(1-t^n)}\frac{\partial} {\partial p_n} {u^{n}}\right)=
%\exp \left(\mp\sum_{i=0}^\infty \sum_{n\ge 1}\frac{\partial} {\partial p_n} {(t^iu)^{n}}\right)=
\prod_{i=0}^{\infty}\exp \left(\mp \sum_{n\ge 1}\frac{\partial} {\partial p_n} {(t^iu)^{n}}\right)
=\prod_{i=0}^{\infty}E^\perp(-u/t^i)
\end{align*}
Note that 
\[
\exp \left(\sum_{n\ge 1}\frac{(1-t^n)p_n}{n} \frac{1}{u^n}  \right)= H(u) E(-u/t),%\quad \exp \left(-\sum_{n\ge 1}\frac{(1-t^n)p_n}{n} \frac{1}{u^n}  \right)= E(−u)H(u/t).
\]
which gives (\ref{phit1}) for $\Phi_t^+(u)$, and similarly (\ref{phit2}) for $ \Phi_t^-(u)$.

Applying the same substitution   $p_n\to (1-t^n)p_n$  
to the result of Proposition  \ref{v_pres_s} we get c).

\end{proof}

\end{document}